\documentclass[runningheads,a4paper]{llncs}

\usepackage[latin1]{inputenc}
\usepackage{amssymb,amsmath}

\usepackage{url}
\urldef{\mailsa}\path|aaciss@ept.sn|
\urldef{\arxiv}\path |http://www.arxiv.org|
\newcommand{\keywords}[1]{\par\addvspace\baselineskip
\noindent\keywordname\enspace\ignorespaces#1}

\begin{document}

\title{Two-sources Randomness Extractors for Elliptic Curves}
\titlerunning{Two-sources Randomness Extractors for Elliptic Curves}

\author{Abdoul Aziz Ciss}
\authorrunning{Abdoul Aziz Ciss }
\institute{Laboratoire de Traitement de l'Information et Syst\`emes Intelligents,
\\ \'Ecole Polytechnique de Thi\`es, S\'en\'egal \\
\mailsa}
\maketitle
\bibliographystyle{plain}

\begin{abstract}
This paper studies the task of  two-sources randomness extractors for elliptic curves defined over a finite field $K$, where $K$ can be a prime or a binary field. In fact, we  introduce new constructions of functions over elliptic curves which take in input two random points from two different subgroups. In other words, for a given elliptic curve $E$ defined over a finite field $\mathbb{F}_q$ and two random points $P \in \mathcal{P}$ and $Q\in \mathcal{Q}$, where $\mathcal{P}$ and $\mathcal{Q}$ are two subgroups of  $E(\mathbb{F}_q)$, our function extracts the least significant bits of the abscissa of the point $P\oplus Q$ when $q$ is a large prime, and the $k$-first $\mathbb{F}_p$ coefficients of the abscissa of the point $P\oplus Q$ when $q = p^n$, where $p$ is a prime greater than $5$. We show that the extracted bits are close to uniform.

Our construction extends some interesting randomness extractors for elliptic curves, namely those defined in \cite{op} and \cite{ciss1,ciss2}, when $\mathcal{P} = \mathcal{Q}$. The proposed constructions can be used  in any cryptographic schemes which require extraction of random bits from two sources over elliptic curves, namely in key exchange protocol , design of strong pseudo-random number generators, etc.

\end{abstract}

\section{Introduction}
A deterministic randomness extractor  for an elliptic curve is a function which allows to produce close to uniform random bit-string from a random point of the elliptic curve. The main difficulty of extracting randomness in elliptic curve points is to find suitable and explicit constructions for such function, ie. computable in polynomial time by a Turing Machine.

The task of randomness extraction from a point of an elliptic curve has several cryptographic applications. For example, it can
be used in key derivation functions, in key exchange protocols like Diffie-Hellman \cite{dh} and to design  cryptographically secure pseudorandom number generators \cite{Trevisan}.

For instance, by the end of Diffie-Hellman key exchange protocol \cite{dh}, Alice and Bob agree on a common secret $K_{AB} \in G$,
where $G$ is a cryptographic cyclic group,  which is indistinguishable from another element of $G$ under the decisional Diffie-Hellman assumption \cite{ddh}. The secret key used for encryption or authentication of data has to be indistinguishable from a uniformly random bit-string. Hence, the common secret $K_{AB}$ cannot be directly used as a session key.

A classical  solution is the use of a hash function to map an element of the group $G$ onto a uniformly random bit-string of fixed length.  However, the indistinguishability cannot be proved under the decisional Diffie-Hellman assumption. In this case,  it is necessary to appeal to the Random Oracle or to other technics. Many results in this direction can be found in  \cite{re,pr}. An alternative to hash function is to use a deterministic extractor when $G$ is the group of points of an elliptic  curve  \cite{op,theseciss,ciss1,ciss2,qe,be,gu}. These constructions use exponential sums to bound the statisticall distance.

In this paper, we introduce two new constructions of two-sources randomness extractors for elliptic curves defined over finite field. More precisely, we deal with  finite fields $\mathbb{F}_p$ for large prime $p$ and finite fields $\mathbb{F}_q$ where $q = p^n$. Consider an elliptic curve $E$ defined over a finite field $\mathbb{F}_p$, with $p > 5$, and $\mathcal{P}$ and $\mathcal{Q}$ be two distinct subgroups of $E(\mathbb{F}_q)$. For given two points $P\in\mathcal{P}$ and $Q\in\mathcal{Q}$, the first extractor outputs the $k$-least significant bits of the abscissa of the point $P\oplus Q$. We show that the extracted bits are indistinguishable from a random bit-string  of length $k$. In fact, we use bilinear exponential sums, recently proposed by Ahmadi and Shparlinski \cite{Ahmadi} to bound the the statistical distance.

We use the same technique to defined a two-source randomness extractor for elliptic curves defined over finite fields $\mathbb{F}_q$, where $q = p^n$. The proposed function extracts the $k$-first $\mathbb{F}_p$ coefficients of the abscissa of the point $P\oplus Q$.

We organize the paper as follows : the next section recalls some basic notion on theory of randomness extraction, namely tools for measuring randomness : collision probability, statistical distance, min-entropy, exponential, character  sums over finite fields and elliptic curves, in particular we recall fundamental results  on bilinear exponential sums over elliptic curves we use in this paper. We also give some previous results related to the randomness extraction in elliptic curves when working only one subgroup. Section 3 introduces our first contribution, ie. a new construction of a two-source deterministic randomness extractor for elliptic curves defined over prime fields. An analogue of this extractor for elliptic curves defined over $\mathbb{F}_{p^n}$ is given in Section 4.

\section{Preliminaries}
\subsection{Deterministic extractor}

\begin{definition} [Collision probability] Let $S$ be a finite set and  $X$ be an $S$-valued random variable. The collision probability of $X$, denoted by $Col(X)$, is the probability
$$Col(X)=\displaystyle\sum_{s\in S}\mathrm{Pr}[X=s]^2$$
If $X$ and $X'$ are identically distributed random variables on $S$, the collision probability of $X$ is interpreted as $Col(X)=\mathrm{Pr}[X=X']$
\end{definition}

\begin{definition}[Statistical distance] Let $X$ and $Y$ be $S$-valued random variables, where $S$ is a finite set. The statistical
distance $\Delta (X,Y)$ between $X$ and $Y$ is
\begin{center}
 $\Delta (X,Y)=\frac{1}{2}\displaystyle\sum_{s\in S} \left|\mathrm{Pr}[X=s]-\mathrm{Pr}[Y=s]\right|$
\end{center}
Let $U_S$ be a random variable uniformly distributed on $S$. Then a random variable $X$ on $S$ is said to be $\delta$-uniform if
\begin{center}
$\Delta (X,Y)\leq \delta$
\end{center}
\end{definition}
An equivalent definition is that $|X(A) - Y(A)| \leq \epsilon$ for every event $A \subseteq S$, which means that the two distributions are almost indistinguishable.
\begin{lemma}
Let $S$ be a finite set and let  $(\alpha_x)_{x\in S}$ be a sequence of real numbers. Then,
\begin{equation}\label{eqcs}
\frac{(\sum_{x\in S}|\alpha_x|)^2}{|S|} \leq \sum_{x\in  S} \alpha_x^2.
\end{equation}
\end{lemma}
\emph{Proof.}
This inequality is a direct consequence of Cauchy-Schwarz inequality:
$$\sum_{x\in S}|\alpha_x| = \sum_{x\in S}|\alpha_x|.1 \leq \sqrt{\sum_{x\in  S} \alpha_x^2}\sqrt{\sum_{x\in  S} 1^2} \leq \sqrt{|S|}\sqrt{\sum_{x\in  S} \alpha_x^2}.$$
The result can be deduced easily.

If  $X$ is an $S$-valued random variable and if we consider that $\alpha_x = \mathrm{Pr}[X=x]$, then
\begin{equation}\label{eqcol}
\frac{1}{|S|} \leq Col(X),
\end{equation}
since the sum of probabilities is 1 and since $Col(X)=\displaystyle\sum_{x\in S}\mathrm{Pr}[X=x]^2$.

The following lemma gives an explicit relation between the statistical distance and collision probability.
\begin{lemma}\label{colision}
Let $X$ be a random variable over a finite $S$ of size $|S|$ and $\delta=\Delta (X,U_S)$ be the statistical distance between $X$
and $U_S$, the uniformly distributed  random variable over $S$. Then,
\begin{center}
$Col(X)\geq \displaystyle\frac{1+4\delta^2}{|S|}$
\end{center}
\end{lemma}
\emph{Proof. }
If $\delta = 0$, then the result is an easy consequence of Equation \ref{eqcol}. Let suppose that $\delta\neq 0$ and  define
$$q_x = |\mathrm{Pr}[X=x] -1/|S||/2\delta. $$
Then $\sum_{x}q_x = 1 $ and by Equation \ref{eqcs}, we have
\begin{align*}
\frac{1}{|S|} & \leq \sum_{x\in S} q_x^2 = \sum_{x\in S} \frac{(\mathrm{Pr}[X=x] - 1/|S|)^2}{4\delta^2} = \frac{1}{4\delta^2}\left(\sum_{x\in S}\mathrm{Pr}[X=x]^2 - 1/|S| \right) \\
& \leq \frac{1}{4\delta^2}(Col(X)-1/|S|).
\end{align*}
The lemma can be deduced easily.

\begin{definition}[Min-entropy]
The min-entropy of a distribution $X$ on a set $S$ denoted by $H_{\infty}(x)$ is defined by :
$$H_{\infty}(x) = \min_{x\in S} \log_2 \frac{1}{\mathrm{Pr}[X = x]}$$
\end{definition}
In other words, a distribution has a min-entropy at least $k$ if the probability of each element is bounded by $2^{-k}$. Intuitively, such a distribution contains $k$ random bits.

\begin{definition}[Extractor]
Let $S$ and $T$ be two finite sets. A $(k, \epsilon)$-extractor is a function $$Ext : S \longrightarrow T $$ such that  for every distribution $X$ on $S$ with $H_{\infty}(x) \geq k$ the distribution $Ext(X)$ is $\epsilon$-close to the uniform distribution on $\{0,1\}^m$
\end{definition}

\begin{definition}[Two-sources-extractor]
Let $R$, $S$ and $T$ be finite sets. The function $Ext : R\times S \longrightarrow T$ is a two-sources-extractor if the distribution $Ext(X_1, X_2)$   is $\delta$-close to the uniform distribution $U_T$ for every uniformly distributed random variables $X_1$ in $R$ and $X_2$  in $S$
\end{definition}
For more information on extractors, see \cite{Nisan1,Nisan2,rd,Shaltiel,ran}.

\subsection{Character sums in finite fields}
In the following, we denote by $e_p$ the character on $\mathbb{F}_p$ such that, for all $x\in \mathbb{F}_p$
\begin{center}
$e_p(x)=e^{\frac{2i\pi x}{p}}$ $\in \mathbb{C}^{*}$.
\end{center}
If $I$ is an interval of integers, it's well known that
$$\sum_{x\in\mathbb{F}_p}\left|\sum_{\theta \in I}e_p(\theta x)\right| \leq p\log_2(p)$$
Denote  by $\Psi=$Hom$(\mathbb{F}_{p^n}, \mathbb{C}^{*})$, the group of additive characters on $\mathbb{F}_{p^n}$ that can be described by the set
\begin{center}
$\Psi=\{\psi, \psi(z)=e_p(\mathrm{Tr}(\alpha z)), \mbox{for } \alpha \in
\mathbb{F}_{p^n}\}$
\end{center}
where $\mathrm{Tr}(x)$ is the trace of $x\in \mathbb{F}_{p^n}$ to
$\mathbb{F}_{p}$ (see \cite{oes}).

\begin{lemma}\label{sums}
Let $V$ be an additive subgroup of $\mathbb{F}_{p^{n}}$. Then,
\begin{equation*}
\sum_{\psi\in \Psi}\left|\sum_{z\in V}\psi(z) \right| \leq p^{n}.
\end{equation*}
\end{lemma}

\begin{proof}
 See \cite{incomplete} for the proof.
\end{proof}

\subsection{Character sums with elliptic curves}
Let $q$ be a prime power and let $E$ be an elliptic curve defined over a finite field $\mathbb{F}_q$ of $q$ elements of characteristic $p\geq 5$ given by an affine Weierstrass equation $$E : y^2 = x^3 + ax + b$$ with $a, b\in \mathbb{F}_q$, see \cite{Avanzi,Blake,Koblitz,Koblitz2,Silverman}. The set of all points on $E$ forms an abelian group with neutral element $\mathcal{O}$. Let $oplus$ denote the group law operation. For a point $P\neq \mathcal{O}$ on $E$ we write $P = (\mathrm{x}(P), \mathrm{y}(P))$. Let $\psi$ be a non principal additive character of $\mathbb{F}_q$ and let $\mathcal{P}$ and $\mathcal{Q}$ be two subsets of  $E(\mathbb{F}_q)$. For arbitrary complex functions $\rho(P)$ and $\vartheta(Q)$ supported on $\mathcal{P}$ and $\mathcal{Q}$ we consider the bilinear sums of additive type:
$$V_{\rho, \vartheta}(\psi, \mathcal{P}, \mathcal{Q}) = \sum_{P\in \mathcal{P}}\sum_{Q\in \mathcal{Q}} \rho(P)\vartheta(Q)\psi(\mathrm{x}(P\oplus Q)).$$
We recall the following interesting result of \cite{Ahmadi}.

\begin{lemma}\label{base}
Let $E$ be an elliptic curve defined over $\mathbb{F}_q$ and let $$\sum_{P\in \mathcal{P}}|\rho(P)|^2 \leq T \quad \text{ and } \quad \sum_{Q\in \mathcal{Q}}|\vartheta(Q)|^2\leq T.$$ Then, uniformly over all nontrivial additive character $\psi$ of $\mathbb{F}_q$,
$$|V_{\rho, \vartheta}(\psi, \mathcal{P}, \mathcal{Q})|\ll \sqrt{qRT}$$
\end{lemma}

\begin{proof}
 See \cite{Ahmadi}
\end{proof}

\subsubsection{Previous works}

For $q = p$ a prime number $>5$ let's recall the extractor of Chevalier \emph{et al.} in \cite{op}
\begin{definition}
Let $E$  be an elliptic curve defined over a finite field $\mathbb{F}_p$, for a prime $p>2$. Let $G$ be a subgroup of $E(\mathbb{F}_p)$ and let $k$ be a positive integer. Define the function
\begin{align*}
\mathcal{L}_k : G& \longrightarrow \{0, 1\}^k\\
				P& \longmapsto \mathrm{lsb}_k(\mathrm{x}(P))
\end{align*}
\end{definition}
The following lemmas state that $\mathcal{L}_k$ is a deterministic randomness extractor for the elliptic curve $E$

\begin{lemma}
Let $p$ be a $n$-bit prime, $G$ a subgroup of $E(\mathbb{F}_p)$ of cardinal $q$ generated by a point $P_0$, $q$ being an $l$-bit prime, $U_G$ a random variable uniformly distributed in $G$ and $k$ a positive integer. Then
$$\Delta(\mathcal{L}_k(U_G), U_k)\leq 2^{(k+n+\log_2(n))/2 + 3 - l},$$
where $U_k$ is the uniform distribution in $\{0, 1\}^k$.
\end{lemma}
~~\\
\emph{Proof. } See \cite{op}.
\begin{corollary}
Let $e$ be a positive integer and suppose that
$$k \leq 2l -(n+2e+\log_2(n)+6).$$ Then $\mathcal{L}_k$ is a $(U_G, 2^{-e})$-deterministic extractor
\end{corollary}
~~\\~~\\Consider now the finite field $\mathbb{F}_{p^n}$, where $p > 5$ is prime and $n$ is a positive integer. Then $\mathbb{F}_{p^n}$ is a $n$-dimensional vector space over $\mathbb{F}_p$. Let $\{\alpha_1, \alpha_2, \ldots, \alpha_n\}$ be a basis of $\mathbb{F}_{p^n}$ over $\mathbb{F}_p$. That means, every element $x$ of $\mathbb{F}_{p^n}$ can be represented in the form $x=x_1\alpha_1 + x_2\alpha_2 +\ldots +x_n\alpha_n$, where $x_i\in\mathbb{F}_{p^n}$.
Let $E$ be the elliptic curve over $\mathbb{F}_{p^n}$ defined by the Weierstrass equation
$$y^2+(a_1 x+a_3)y=x^3+a_2 x^2+a_4 x+ a_6.$$

The extractor $\mathcal{D}_k$, where $k$ is a positive integer less than $n$,  for a given point $P$ on $E(\mathbb{F}_{p^n})$, outputs the $k$ first $\mathbb{F}_p$-coordinates of the abscissa of the point $P$.
\begin{definition}
Let $G$ be a subgroup of $E(\mathbb{F}_{p^n})$ and $k$ a positive integer less than $n$. Define the function $\mathcal{D}_k$
\begin{align*}
\mathcal{D}_k : G & \longrightarrow \mathbb{F}_{p^k}\\P=(x,y) &
\longmapsto (x_1, x_2, \ldots, x_k)
\end{align*}
where $x\in \mathbb{F}_{p^n}$ is represented as $x=x_1\alpha_1 + x_2\alpha_2 +\ldots +x_n\alpha_n$, and $x_i\in\mathbb{F}_{p^n}$.
\end{definition}

\begin{lemma}
Let $E$ be an elliptic curve defined over $\mathbb{F}_{q}$, whit $q = p^n$  and let $G$ be a subgroup of $E(\mathbb{F}_{q})$. Let $\mathcal{D}_k$ be the function defined above. Then,
$$\mathrm{Col}(\mathcal{D}_k(U_G)\leq \frac{1}{p^k} + \frac{4\sqrt{q}}{|G|^2}$$
and
$$\Delta(\mathcal{D}_k(U_G), U_{\mathbb{F}_{p^k}})\leq \displaystyle \frac{2\sqrt{p^{n+k}}}{|G|}$$
where $U_G$ is uniformly distributed in $G$ and $U_{\mathbb{F}_{p^k}}$ is the uniform distribution in
$\mathbb{F}_{p^k}$.
\end{lemma}

\begin{proof}
See \cite{ciss2}
\end{proof}

\begin{lemma}
Let $p>2$ be a prime and  $E(\mathbb{F}_{p^n})$ be an elliptic curve over $\mathbb{F}_{p^n}$ and  $G \subset E(\mathbb{F}_{p^n})$ be a multiplicative subgroup of order $r$ with $|r|=t$ and $|p|=m$ and let $U_G$ be the uniform distribution in $G$. If  $e>1$ is an integer  and $k>1$ is an integer such that
\begin{equation*}
k\leq \frac{2t - 2e - nm - 4}{m},
\end{equation*}
then $\mathcal{D}_k$ is a $(\mathbb{F}_{p}^k, 2^{-e})$-deterministic randomness extractor over the elliptic curve $E(\mathbb{F}_{p^n})$.
\end{lemma}

\begin{proof}
See \cite{ciss2}
\end{proof}


\section{Randomness extractors for  $E(\mathbb{F}_p)$}
\begin{definition}
Let $E$ be an elliptic curve defined a finite field $\mathbb{F}_q$, with $q=p$ a prime greater than 5, and let $\mathcal{P}$ and $\mathcal{Q}$ be  two subgroups of $E(\mathbb{F}_q)$ with $\#\mathcal{P} = r$ and $\#\mathcal{Q} = t$. Define the function
\begin{align*}
Ext_1 : \mathcal{P} \times \mathcal{Q}  & \longrightarrow \{0, 1\}^k\\
					(P, Q)  & \longmapsto \mathrm{lsb}_k(\mathrm{x}(P\oplus Q))
\end{align*}
\end{definition}

\begin{theorem}\label{RECFP}
Let $E$ be an elliptic curve defined over $\mathbb{F}_p$ and let $\mathcal{P}$ and $\mathcal{Q}$ be two subgroups of $E(\mathbb{F}_p)$, with $\#\mathcal{P} = r$ and $\#\mathcal{Q} = t$. Let $U_{\mathcal{P}}$ and $U_{\mathcal{Q}}$ be two random variables uniformly distributed in $\mathcal{P}$ and $\mathcal{Q}$ respectively and let $U_k$ be the uniform distribution in $\{0, 1\}^k$. Then,
$$\Delta(Ext_1(U_{\mathcal{P}}, U_{\mathcal{Q}}), U_k) \ll \sqrt{\frac{2^{k-1}p\log(p)}{rt}}$$
\end{theorem}

\begin{proof}
Let $\alpha = 2^k$ and let $\theta_0 = \mathrm{msb}_{n-k}(p-1)$. Define the set
$$\mathcal{A} = \{(P, Q), (R, S) \in \mathcal{P}\times \mathcal{Q}\  | \ \exists \ \theta \leq \theta_0, \mathrm{x}(P\oplus Q) - \mathrm{x}(R\oplus S) - \alpha \theta = 0 \ \mathrm{mod} \ p\}.$$
Consider the double character sum $V_{\rho, \vartheta}(\psi, \mathcal{P}, \mathcal{Q})$, with $\rho(P) = 1\quad \forall\  P$ and $\vartheta(Q) = 1 \quad \forall \ Q$.
Then,
\begin{align*}
\mathrm{Col}(Ext_1(U_{\mathcal{P}}, U_{\mathcal{Q}})) &= \frac{\#\mathcal{A}}{(rt)^2}\\
													  &=\frac{1}{r^2t^2p}\sum_{P\in \mathcal{P}}\sum_{Q\in \mathcal{Q}}\sum_{R\in \mathcal{P}}\sum_{S\in \mathcal{Q}}\sum_{\theta\leq \theta_0}\sum_{\psi \in \Psi}\psi(\mathrm{x}(P\oplus Q) - \mathrm{x}(R\oplus S) - \alpha \theta)\\
													  &=\frac{1}{2^k}+\frac{1}{r^2t^2p}\sum_{P\in \mathcal{P}}\sum_{Q\in \mathcal{Q}}\sum_{R\in \mathcal{P}}\sum_{S\in \mathcal{Q}}\sum_{\theta \leq \theta_0}\sum_{\psi \neq \psi_0}\psi(\mathrm{x}(P\oplus Q) - \mathrm{x}(R\oplus S) - \alpha \theta)\\
													  &\leq \frac{1}{2^k}+\frac{1}{r^2t^2p}\left|\sum_{P\in \mathcal{P}}\sum_{Q\in \mathcal{Q}}\psi(\mathrm{x}(P\oplus Q))\right| \left|\sum_{R\in \mathcal{P}}\sum_{S\in \mathcal{Q}}\psi(-\mathrm{x}(R\oplus S)\right|\left|\sum_{\theta\leq \theta_0}\sum_{\psi \neq \psi_0} \psi( - \alpha \theta))\right|\\
													  &\ll \frac{1}{2^k}+\frac{V^2}{r^2t^2p}\sum_{\theta\leq \theta_0}\left|\sum_{\psi \neq \psi_0} \psi( - \alpha \theta))\right|\\
													  &\ll \frac{1}{2^k} + \frac{p\log(p)}{rt}
\end{align*}
Therefore,
$$\Delta(Ext_1(U_{\mathcal{P}}, U_{\mathcal{Q}}), U_k) \ll \sqrt{\frac{2^{k-1}p\log(p)}{rt}}$$

\end{proof}

\begin{corollary}
Let $m$ and $l$ be the bit size of $r$ and $t$ respectively and let $e$ be a positive integer. If $k$ is a positive integer such that
$$k\leq m + l -(n+2e+\log_2(n) + 1),$$ then $Ext_1$ is a $(k, O(2^{-e}))$-deterministic extractor for $\mathcal{P}\times \mathcal{Q}$.
\end{corollary}
The following corollary is a generalization of the results of Chevalier \emph{et al.} in \cite{op}.

\begin{corollary}
If $\mathcal{P} = \mathcal{Q}$ and $e$ is a positive integer such that $$k\leq 2l -(n + 2e + \log_2(n) + 1),$$ then $Ext_1$ is a $(k, O(2^{-e}))$-deterministic randomness extractor for $\mathcal{P}$ and generalizes the result of Corollary 15 of \cite{op}.
\end{corollary}

\begin{proof}
\begin{enumerate}
\item In fact, if $\mathcal{P} = \mathcal{Q}$ then $m = l$ and $$k\leq 2l -(n + 2e + \log_2(n)  + 1)$$ for $e  > 0$. Thus, $Ext_1$ is a $(k, O(2^{-e}))$-deterministic randomness extractor for $\mathcal{P}$.
\item Note that if $\mathcal{P} = \mathcal{Q}$, then $Ext_1(P, P)$ for $P\in_R \mathcal{P}$ is equivalent to $\mathcal{L}_k(2P)$. Since the point $2P$ is also random, we have
    $$\Delta(Ext_1(U_{\mathcal{P}}, U_{\mathcal{P}}), U_k) = \Delta(\mathcal{L}_k(U_{\mathcal{P}}), U_k) = O(2^{-e}) $$
\end{enumerate}
\end{proof}

\section{Randomness Extractor for $E(\mathbb{F}_{p^n})$, with $p > 5$}
\begin{definition}
Let $E$ be an elliptic curve defined over the finite field $\mathbb{F}_{p^n}$, where $p$ is a prime greater than 5 and $n > 1$. Consider two subgroups $\mathcal{P}$ and $\mathcal{Q}$ of $E(\mathbb{F}_q)$. Define the function
\begin{align*}
Ext_2 : \mathcal{P} \times \mathcal{Q}  & \longrightarrow \mathbb{F}_p^k\\
					(P, Q)  & \longmapsto (x_1, x_2, \ldots, x_k)
\end{align*}
where $\mathrm{x}(P\oplus Q) = (x_1, x_2, \ldots, x_k, x_{k+1}, \ldots, x_n)$. In other words, the function $Ext_2$ output the $k$ first $\mathbb{F}_p$-coefficients of the point $P\oplus Q$.
\end{definition}

\begin{theorem}
Let $E$ be an elliptic curve defined over $\mathbb{F}_{p^n}$ and let $\mathcal{P}$ and $\mathcal{Q}$ be two subgroup of $E(\mathbb{F}_{p^n})$ with $\#\mathcal{P} = r$ and $\#\mathcal{Q} = t$. Denote by $U_{\mathcal{P}}$ and $U_{\mathcal{Q}}$ two random variables uniformly distributed on $\mathcal{P}$ and $\mathcal{Q}$ respectively. Then,
$$\Delta (Ext_2(U_{\mathcal{P}}, U_{\mathcal{Q}}), U_{\mathbb{F}_p^k}) \ll \sqrt{\frac{p^{n+k}}{4rt}}$$
\end{theorem}
~~\\
\emph{Sketch of proof.} Consider the sets
$$\mathcal{M} = \{(x_{k+1}\alpha_{k+1}+ x_{k+2}\alpha_{k+2}+\ldots+x_n\alpha_n) , x_i \in \mathbb{F}_{p}\}\subset \mathbb{F}_{p^n}$$ and
$$\mathcal{A} = \{(P, Q), (R, S) \in \mathcal{P}\times \mathcal{Q}\quad|\quad \exists \lambda \in \mathcal{M}, \mathrm{x}(P\oplus Q) - \mathrm{x}(R\oplus S) = \lambda\}.$$
Then,
$$\mathrm{Col}(Ext_2(U_\mathcal{P}, U_{\mathcal{Q}})) = \frac{\#\mathcal{A}}{(rt)^2}.$$
Use the technique of the proof of Theorem \ref{RECFP} and  Lemma \ref{sums} and \ref{base} to complete the proof.

\begin{corollary}
Let $p > 5$ be a prime and $E$ be an elliptic curve defined over $\mathbb{F}_{p^n}$, for $n > 0$. Let $\mathcal{P}$ and $\mathcal{Q}$ be two subgroups of $E(\mathbb{F}_{p^n})$, with $r = \#\mathcal{P}$, $t = \#\mathcal{Q}$. Note by $m = |p|$, $l = |t|$ and $s = |r|$. If $e$ is a positive integer such that
$$k \leq \frac{l + s - 2e -mn}{m},$$ then $Ext_2$ is a $(k, O(2^{-e}))$-deterministic randomness extractor for $\mathcal{P} \times \mathcal{Q}$.
\end{corollary}

The following corollary states the equivalence of $Ext_2$ with $\mathcal{D}_k$ when $\mathcal{P} = \mathcal{Q}$.

\begin{corollary}
If $\mathcal{P} = \mathcal{Q}$ and $e$ is positive integer such that $$k \leq \frac{2l - 2e - mn}{m},$$ then $Ext_2$ is a $(k, O(2^{-e}))$-deterministic randomness extractor for $\mathcal{P}$ and generalizes the randomness extractor of Ciss \emph{et al.}
\end{corollary}


\end{document}